\documentclass[pdftex,11pt,a4paper]{article}
\usepackage{amsmath,amssymb,amsfonts,amsthm}
\usepackage{algorithmicx,algorithm,algpseudocode}
\usepackage{longtable}
\usepackage{enumerate}
\usepackage{graphicx}
\usepackage{algpseudocode}
\usepackage{algorithm}
\usepackage[title]{appendix}
\usepackage{stmaryrd}
\usepackage{float}
\usepackage[font=small,labelfont=bf]{caption}
\usepackage[left=3cm, right=3cm, top=2.5cm,  bottom=2.5cm ]{geometry}
\newtheorem{theorem}{Theorem}
\newtheorem{lemma}{Lemma}
\newtheorem{definition}{Definition}

\newtheorem{corollary}{Corollary}

\newtheorem{remark}{Remark}                               
\newtheorem*{remark*}{Remark} 				
\newtheorem{example}{Example}                             
\newtheorem*{example*}{Example}                          
\newtheorem*{note*}{Note} 				

\DeclareMathOperator{\Diag}{Diag} 
\newtheorem{fact}{\bf Fact}
\usepackage{enumerate}
\usepackage{hyperref}
\hypersetup{
	colorlinks   = true, 
	urlcolor     = red, 
	linkcolor    = blue, 
	citecolor   = blue 
}
\newcommand{\splitatcommas}[1]{%
	\begingroup
	\begingroup\lccode`~=`, \lowercase{\endgroup
		\edef~{\mathchar\the\mathcode`, \penalty0 \noexpand\hspace{0pt plus 1em}}%
	}\mathcode`,="8000 #1%
	\endgroup
}
\usepackage[
n,
operators,
advantage,
sets,
adversary,
landau,
probability,
notions,	
logic,
ff,
mm,
primitives,
events,
complexity,
asymptotics,
keys]{cryptocode}

\begin{document}
	\title{Construction of all MDS and involutory MDS matrices}
	\author{Yogesh Kumar\footnote{Scientific Analysis Group, DRDO, Metcalfe House Complex, Delhi-110054\newline email: \texttt{adhana.yogesh@gmail.com}}, 
		P.R.Mishra\footnote{Scientific Analysis Group, DRDO, Metcalfe House Complex, Delhi-110054\newline email: \texttt{prasanna.r.mishra@gmail.com}},
		Susanta Samanta\footnote{R. C. Bose Centre for Cryptology and Security, Indian Statistical Institute, Kolkata-700108\newline email: \texttt{susanta.math94@gmail.com}},
		Kishan Chand Gupta\footnote{Applied Statistics Unit, Indian Statistical Institute, Kolkata-700108\newline email: \texttt{kishan@isical.ac.in}},
		Atul Gaur\footnote{Department of Mathematics, University of Delhi, Delhi-110007 \newline email: \texttt{gaursatul@gmail.com}}}
	\date{}
	\linespread{1.0}
	\maketitle

\begin{abstract}\noindent
In this paper, we propose two algorithms for a hybrid construction of all $n\times n$ MDS and involutory MDS matrices over a finite field $\mathbb{F}_{p^m}$, respectively. The proposed algorithms effectively narrow down the search space to identify $(n-1) \times (n-1)$ MDS matrices, facilitating the generation of all $n \times n$ MDS and involutory MDS matrices over $\mathbb{F}_{p^m}$. To the best of our knowledge, existing literature lacks methods for generating all $n\times n$ MDS and involutory MDS matrices over $\mathbb{F}_{p^m}$. In our approach, we introduce a representative matrix form for generating all $n\times n$ MDS and involutory MDS matrices over $\mathbb{F}_{p^m}$. The determination of these representative MDS matrices involves searching through all $(n-1)\times (n-1)$ MDS matrices over $\mathbb{F}_{p^m}$. Our contributions extend to proving that the count of all $3\times 3$ MDS matrices over $\mathbb{F}_{2^m}$ is precisely $(2^m-1)^5(2^m-2)(2^m-3)(2^{2m}-9\cdot 2^m+21)$. Furthermore, we explicitly provide the count of all $4\times 4$ MDS and involutory MDS matrices over $\mathbb{F}_{2^m}$ for $m=2, 3, 4$. \\


\end{abstract}


\section{Introduction}\label{Sec:Introduction}
Claude Shannon, in his paper ``Communication Theory of Secrecy Systems''~\cite{shan}, introduced the concepts of confusion and diffusion, which play a crucial role in the design of symmetric key cryptographic primitives. In general, the diffusion property is attained through the use of a linear layer, which can be represented as a matrix. This matrix is designed to produce a significant alteration in the output for a small change in the input. The strength of the diffusion layer is usually assessed by its branch number~\cite{JDA_Thesis_1995} and the optimal branch number is achieved by the use of MDS matrices. An application is in the MixColumns operation of AES~\cite{AES}. Furthermore, besides block ciphers, MDS matrices are also widely used in many other cryptographic primitives, such as hash functions (Maelstrom~\cite{MAELSTROM}, Gr$\phi$stl~\cite{GROSTL}, PHOTON~\cite{PHOTON}) and stream ciphers (MUGI~\cite{MUGI}). As a result, the effectiveness of MDS matrices in diffusion layers is widely recognized, and several techniques have been proposed for designing MDS matrices.

In general, there are three techniques for generating MDS matrices.  The first method, known as direct construction, utilizes algebraic methods to generate an MDS matrix of any order without the need for any searching. The second one is the search-based construction. The third method integrates the aforementioned strategies into a hybrid approach. Direct constructions are primarily obtained from Cauchy and Vandermonde matrices, as well as their generalizations. Examples from the literature include references~\cite{Gupta2023direct, kc2, LACAN2003, Roth1989, sdm}.

While direct construction methods provide the feasibility of obtaining MDS matrices of any order, there is no guarantee of achieving a matrix with the optimal hardware area. This holds even for smaller sizes. The second technique, known as the search-based technique, is currently the only known method that can provide an optimal MDS matrix in terms of area. However, this approach is feasible only when the matrix size is small and the field size is not too large. In the context of search techniques, various matrix structures, including circulant, left-circulant, Hadamard, and Toeplitz matrices, are used. Significant work has been done in this direction, as demonstrated in \cite{GR15,CYCLICM,psa,XORM,skf}.

In the third method, referred to as the hybrid approach, a representative MDS matrix is typically obtained through a search method. Subsequently, this representative matrix is utilized to generate multiple MDS matrices. We will now briefly discuss some previous research in this area and discuss the inspiration behind this paper. In \cite{psa}, the authors proposed a hybrid technique for efficiently generating $2^n\times 2^n$ MDS matrices. Their approach introduces a new matrix form called GHadamard, which can be utilized to generate new (involutory) MDS matrices from a Hadamard (involutory) MDS matrix as a representative matrix. Additionally, in \cite{Saka}, a technique was introduced to generate new $n\times n$ MDS matrices that are isomorphic to existing ones. In a recent study \cite{Tuncay23}, a hybrid method is proposed to generate all $4 \times 4$
involutory MDS matrices over $\mathbb{F}_{2^m}$.

To the best of our knowledge, there is currently no efficient method for generating all $n\times n$ MDS and involutory MDS matrices over $\mathbb{F}_{p^m}$. This serves as our motivation to delve into the theoretical aspects of generating all $n \times n$ MDS matrices, as well as involutory MDS matrices over $\mathbb{F}_{p^m}$, and to propose hybrid strategies for achieving this.

We introduce two algorithms designed to generate all $n \times n$ MDS and involutory MDS matrices over $\mathbb{F}_{p^m}$, respectively. Specifically, to generate all MDS matrices over $\mathbb{F}_{p^m}$ of order $n \times n$, we construct representative MDS matrices of order $n$ and demonstrate that they can be obtained by searching through all $(n-1) \times (n-1)$ MDS matrices over $\mathbb{F}_{p^m}$. Compared to the exhaustive search, which necessitates examining all $n \times n$ matrices, our approach is considerably more effective. Additionally, to generate all involutory MDS matrices of order $n \times n$ over $\mathbb{F}_{p^m}$, we provide a necessary and sufficient condition for the representative MDS matrices to ensure that the resulting MDS matrix is always involutory.

In \cite{gmt}, the authors present an explicit formula to enumerate all $3\times3$ involutory MDS matrices over $\mathbb{F}_{2^m}$. However, a precise formula to compute the number of all $3\times3$ MDS matrices over a given finite field $\mathbb{F}_{2^m}$ is not available in the literature. In this paper, we establish that the count of all $3\times 3$ MDS matrices over $\mathbb{F}_{2^m}$ is given by $(2^m-1)^5(2^m-2)(2^m-3)(2^{2m}-9\cdot 2^m+21)$. Moreover, we provide the enumeration of all $4\times 4$ MDS and involutory MDS matrices over $\mathbb{F}_{2^m}$ for $m=2,3,4$. We also deduce the necessary and sufficient conditions on the representative matrix for generating all $n \times n$ MDS and involutory MDS matrices when $n=2,3,4$.

The paper is organized as follows. The mathematical background and notations used in the work are briefly discussed in Section~\ref{Sec:Definiotion} of the article. Section~\ref{Sec:proposed_method} provides the theory underlying the generation of all MDS and involutory MDS matrices over $\mathbb{F}_{p^m}$. In Section~\ref{Sec:Counting_3_MDS}, we prove the formula for counting all MDS matrices of order $3$, and in Section~\ref{Sec:Counting_4_MDS}, we present the count of all $4\times 4$ MDS and involutory MDS matrices over certain finite fields of characteristics $2$. Finally, we conclude in Section~\ref{Sec:Conclusion}.

\section{Mathematical preliminaries}\label{Sec:Definiotion}
In this section, we discuss some definitions and mathematical preliminaries that are important in our context. Let $\mathbb{F}_{p^m}$ be a finite field of order $p^m$, where $p$ is a prime and $m$ is a positive integer. Let $\mathbb{F}_{p^m}^*$ denote the multiplicative group of the finite field $\mathbb{F}_{p^m}$. The set of vectors of length $n$ with entries from the finite field $\mathbb{F}_{p^m}$ is denoted by $\mathbb{F}_{p^m}^n$.

Let $M$ be any $n\times n$ matrix over  $\mathbb{F}_{p^m}$  and let $|M|$ denote the determinant of  $M$. A matrix $D$ of order $n$ is said to be diagonal if $(D)_{i,j}=0$ for $i\neq j$. Using the notation $d_i = (D)_{i,i}$, the diagonal matrix $D$ can be represented as $\Diag(d_1, d_2, \ldots, d_n)$. It is evident that the determinant of $D$ is given by $|D| = \prod_{i=1}^{n} d_i$. Therefore, the diagonal matrix $D$ is non-singular over $\mathbb{F}_{p^m}$ if and only if $d_i \neq 0$ for $1 \leq i \leq n$.

An MDS matrix finds practical applications as a diffusion layer in cryptographic primitives. The concept of the MDS matrix comes from coding theory, specifically from the realm of maximum distance separable (MDS) codes. An $[n, k, d]$ code is MDS if it meets the singleton bound $d = n-k + 1$.

\vspace{4pt}
\begin{theorem}~\cite[page 321]{FJ77}
An $[n, k, d]$ code $C$ with generator matrix $G = [ I ~|~ M ]$, where $M$ is a $k \times ( n - k )$ matrix, is MDS if and only if every square sub-matrix (formed from any $i$ rows and any $i$ columns, for any $i = 1, 2,\ldots, min \{k, n - k \}$) of $M$ is non-singular.
\end{theorem}

\begin{definition}
A matrix $M$ of order $n$ is said to be an MDS matrix if $[I~|~M]$ is a generator matrix of a $[2n,n]$ MDS code.
\end{definition}

Another way to define an MDS matrix is as follows:

\begin{fact}
A square matrix $M$ is an MDS matrix if and only if every square sub-matrix of $M$ is non-singular.
\end{fact}

 One of the elementary row operations on matrices is multiplying a row of a matrix by a non-zero scalar. MDS property remains invariant under such operations. Thus, we have the following result regarding MDS matrices.

\begin{lemma}\cite{kcz}\label{Lemma_DMD_MDS}
Let $M$ be an MDS matrix, then for any non-singular diagonal matrices $D_1$ and $D_2$, $D_1MD_2$ will also be an MDS matrix.
\end{lemma}

 Using involutory diffusion matrices is more beneficial for implementation since it allows the same module to be utilized in both encryption and decryption phases.

\begin{definition}
An involutory matrix is defined as a square matrix $M$ that fulfills the condition $M^2 = I$ or, equivalently, $M = M^{-1}$.
\end{definition}

Now, we recall the concept of quadratic residues \cite{lidl} in a finite field.  A non-zero element $s\in \mathbb{F}_{p^m}$ is said to be a quadratic residue in $\mathbb{F}_{p^m}$ if it is a square of a non-zero element in $\mathbb{F}_{p^m}$. In other words, if there exists some $S\in \mathbb{F}_{p^m}^*$ such that $S^2= s$, then $s$ is a quadratic residue in $\mathbb{F}_{p^m}$. We denote the set of all quadratic residues in $\mathbb{F}_{p^m}$ with $QR(p^m)$. It is known that in a finite field of characteristic $2$, every non-zero element is a quadratic residue \cite{lidl}. Furthermore, if $p$ is an odd prime, then it is known that exactly $(p^m-1)/2$ elements in $\mathbb{F}_{p^m}$ are quadratic residues.\\

\section{The proposed methods}\label{Sec:proposed_method}
In this section, we introduce a technique for generating all $n \times n$ MDS and involutory MDS matrices over $\mathbb{F}_{p^m}$. The proposed approach is a hybrid construction that combines search-based and direct construction techniques. The fundamental idea behind the suggested approach is first to identify representative MDS matrices of order $n$ using the search-based method and then to acquire all $n \times n$ MDS and involutory MDS matrices by multiplying two non-singular diagonal matrices (pre and post) with these representative MDS matrices of order $n$. To find all representative MDS matrices of order $n$ over $\mathbb{F}_{p^m}^*$, we define the representative matrix form $M_1$ as follows:
\begin{equation}\label{1}
 M_1=\begin{pmatrix}
1&1&\ldots&1\\
1& & &\\
\vdots & &R&\\
1&  & &
\end{pmatrix},
\end{equation}
where $R$ is a $(n-1) \times (n-1)$ matrix over $\mathbb{F}_{p^m}\setminus\{0,1\}$.

It is evident that the matrix form $M_1$ provided above for obtaining representative MDS matrices of order $n$ is constructed from the $(n-1) \times (n-1)$ MDS matrices $R$ over $\mathbb{F}_{p^m} \setminus \set{0,1}$. Consequently, the search space for finding the representative MDS matrices $M_1$ of order $n$ is drastically reduced as compared to the exhaustive search of finding all $n \times n$ MDS matrices over $\mathbb{F}_{p^m}$.

In the following theorem, we demonstrate that a matrix $M$ over $\mathbb{F}_{p^m}^*$ can be uniquely represented as $M=D_1M_1D_2$, where $D_1$ and $D_2$ are some non-singular diagonal matrices, and $M_1$ is a matrix of the form shown in (\ref{1}).

 \begin{theorem}\label{thm:1}
 Let $M=(a_{i,j})$ be a $n \times n$ matrix over $\mathbb{F}_{p^m}^*$. Then, there are unique $n\times n$ matrices $D_1, D_2,$ and $M_1$ over $\mathbb{F}_{p^m}^*$ such that $$M=D_1M_1D_2,$$ where $D_1$ and $ D_2$ are diagonal matrices, $D_2$'s first entry is $1$ and $M_1$ is the representative matrix.
\end{theorem}
\begin{proof}
Let us choose  $D_1=\Diag(a_{1,1},a_{2,1},\ldots,a_{n,1})$ and $D_2=\Diag(1,a_{1,1}^{-1}$ $a_{1,2},\ldots,\\a_{1,1}^{-1}a_{1,n})$.
Since they have entries in $\mathbb{F}_{p^m}^*$, $D_1$ and $D_2$ are non-singular. We consider $M_1=D_1^{-1}MD_2^{-1}$. It is therefore simple to see that $M=D_1M_1D_2$. It must be demonstrated that $M_1$ has the desired form as stated in (\ref{1}).
Let $M_1=(c_{i,j})$, where $c_{i,j}\in \mathbb{F}_{p^m}^*$. Then, we note that
\[c_{i,j}=a_{i,1}^{-1}a_{i,j}(a_{1,1}^{-1}a_{1,j})^{-1}=a_{i,1}^{-1}a_{i,j}a_{1,1}a_{1,j}^{-1}.\]
Furthermore, we can observe that
\begin{equation}\label{eq:1.5}
c_{i,1}=c_{1,i}=1,\ \text{for all}\  i=1,2,\ldots,n.
\end{equation}
Therefore, $M_1$ has the desired form.

Next, we prove the uniqueness part. Let there be two triplets $(D_1,M_1,D_2)$ and $(D_1',M_1',D_2')$ for the same $M$ that meet the theorem's requirements. Let $D_1=\Diag(\lambda_1,\lambda_2,\\ \ldots,\lambda_n)$, $D_2=\Diag(1,\theta_2,\ldots,\theta_n)$ and
$D_1'=\Diag(\lambda_1',\lambda_2',
\ldots,\lambda_n')$, $D_2'=\Diag(1,\theta_2',\ldots,\theta_n')$ for $\lambda_i, \theta_i, \lambda_i', \theta_i'\in \mathbb{F}_{p^m}^*$ for each $i=1, 2, \ldots, n$.
Also, let \smallskip
\[M_1=\begin{pmatrix}
1 &1 &\ldots &1\\
1 &r_{2,2} &\ldots &r_{2,n}\\
1 &\ldots &\ddots &\ldots\\
1 &r_{n,2} &\ldots &r_{n,n}
\end{pmatrix}\ \ \text{and}\ \
M'_1=\begin{pmatrix}
1 &1 &\ldots &1\\
1 &r'_{2,2} &\ldots &r'_{2,n}\\
1 &\ldots &\ddots &\ldots\\
1 &r'_{n,2} &\ldots &r'_{n,n}
\end{pmatrix}.\]
Since $D_1M_1D_2=D'_1M'_1D'_2$, we note that 
$$\begin{pmatrix}	\lambda_1 &\lambda_1\theta_2 &\ldots&\lambda_1\theta_n\\
\lambda_2 &\lambda_2\theta_2r_{2,2} &\ldots &\lambda_2\theta_nr_{2,n}\\
\ldots&\ldots &\ddots &\ldots\\
\lambda_n &\lambda_n\theta_2r_{n,2} &\ldots &\lambda_n\theta_nr_{n,n}
\end{pmatrix}=\begin{pmatrix}
\lambda_1' &\lambda'_1\theta'_2 &\ldots&\lambda'_1\theta'_n\\
\lambda'_2 &\lambda'_2\theta'_2r'_{2,2} &\ldots &\lambda'_2\theta'_nr'_{2,n}\\
\ldots&\ldots &\ddots &\ldots\\
\lambda'_n &\lambda'_n\theta'_2r'_{n,2} &\ldots &\lambda'_n\theta'_nr'_{n,n}\end{pmatrix}.
$$
A direct comparison of the elements of the matrices above reveals that $$D_1=D'_1,\ D_2=D'_2\ \text{and}\ M_1=M'_1.$$
So, uniqueness follows. The proof is now complete.
\end{proof}

Let's define the set $\mathcal{S}_n(\mathbb{F}_{p^m})$  as the collection of triplets of matrices $(D_1, D_2, M_1)$ satisfying the conditions given out in the Theorem \ref{thm:1}. The set of all $n\times n$ matrices over $\mathbb{F}_{p^m}^*$ is defined as the set $\mathcal{T}_n(\mathbb{F}_{p^m})$. Theorem \ref{thm:1} associates every element of $\mathcal{T}_n(\mathbb{F}_{p^m})$ to a triplet of type $(D_1, D_2, M_1)$, where $D_1$ and $D_2$ are diagonal matrices and $M_1$ is a representative matrix. It also asserts that all such triplets are distinct. These facts allow us to easily conclude that the map $(D_1,D_2,M_1) \mapsto M$ puts the sets $\mathcal{S}_n(\mathbb{F}_{p^m})$  and $\mathcal{T}_n(\mathbb{F}_{p^m})$  in one-to-one correspondence. Let's refer to this correspondence as $\Phi$.

\begin{remark}
It is worth mentioning that Gupta et al. in~\cite[Remark~3]{kc2} presented a notion of maximizing the occurrences of $1$'s in an MDS matrix, and the obtained MDS matrix with $2n-1$ occurrences of $1$ is of the same form as $M_1$ as shown in (\ref{1}). In this paper, we have chosen $M_1$ as the representative matrix form to generate all $n\times n$ MDS and involutory MDS matrices. Also, we demonstrate that for a matrix $M$, there is unique tuple $(D_1,D_2,M_1)$ such that $\Phi(D_1,D_2,M_1)=M$.
\end{remark}

\begin{remark}\label{Remark:M1_not_involutory}
It is important to note that the matrices of the form $M_1$ as shown in (\ref{1}) can never be involutory~\cite[Remark~3]{kc2}. Therefore, we consider it a representative of a class of matrices and investigate the conditions under which $M_1$ yields an involutory matrix $\Phi(D_1, D_2, M_1)$.
\end{remark}

The conditions under which the matrix $\Phi(D_1, D_2, M_1)$ will be an involutory matrix over $\mathbb{F}_{p^m}^*$ are discussed in the following theorem. It also discusses the particular forms of $D_1$ and $D_2$ when $\Phi(D_1, D_2, M_1)$ is involutory.

\begin{theorem}\label{thm:2}
Let $M_1=(c_{i,j})$ be a representative matrix of order $n$ and $M_2=(d_{i,j})$ be its inverse. Then $\Phi(D_1,D_2,M_1)$ will be an involutory matrix if and only if $\exists~ \alpha_i \in \mathbb{F}_{p^m}^{*}$ such that
\begin{center}
$d_{i,j}=\alpha_i \alpha_j c_{i,j}$, $1\leq i,j\leq n$.
\end{center}
Moreover, when $\Phi(D_1,D_2,M_1)$ is involutory, $D_1$ and $D_2$ must take the following form:
\begin{center}
$D_1=\Diag(\alpha_1,\lambda_2,\lambda_3,\ldots,\lambda_n)$ and $D_2=\Diag(1,\frac{\alpha_2}{\lambda_2},\frac{\alpha_3}{\lambda_3},\ldots,\frac{\alpha_n}{\lambda_n})$,
\end{center}
where $\lambda_i \in \mathbb{F}_{p^m}^{*}$ for $i=2,3,\ldots,n$.
\end{theorem}

\begin{proof}
Let $D_1=\Diag(\lambda_1,\lambda_2,\ldots,\lambda_n)$ and $D_2=\Diag(1,\theta_2,\ldots,\theta_n)$ be two non-singular matrices over $\mathbb{F}_{p^m}^{*}$. Suppose that $\Phi(D_1,D_2,M_1)=M$ is an involutory matrix. Therefore, we have
\begin{equation*}
\begin{aligned}
& M^2 = I \\
\implies & (D_1M_1D_2)(D_1M_1D_2) = I \\
\implies & M_1 (D_2D_1M_1D_2) = D_1^{-1} \\
\implies & M_1 (D_2D_1M_1D_2D_1) = I.
\end{aligned}
\end{equation*}
	
\noindent Now, since $M_2$ is the inverse of $M_1$, from above we can say that
\begin{center}
$M_2=D_2D_1M_1D_2D_1$.
\end{center}
	
\noindent Now since $D_2D_1=\Diag(\lambda_1,\lambda_2 \theta_2,\ldots,\lambda_n \theta_n)$, we have
\begin{equation}\label{Th2_Eqn_1}
\begin{aligned}
d_{i,j}= \lambda_i \theta_i c_{i,j}  \lambda_j \theta_j,
\end{aligned}
\end{equation}
for $1\leq i,j \leq n$ and $\theta_1=1$. Now if we compare the diagonal entries, we will have
\begin{equation}\label{Th2_Eqn_2}
\begin{aligned}
d_{i,i}= \lambda_i^2 \theta_i^2 c_{i,i},
\end{aligned}
\end{equation}
for $1\leq i \leq n$.

Now, let $\lambda_i^2 \theta_i^2=\alpha_i^2$ for some $\alpha_i \in \mathbb{F}_{p^m}^{*}$. Therefore, from   (\ref{Th2_Eqn_2}), we will have
\begin{equation}\label{3.5}
d_{i,i}= \alpha_i^2 c_{i,i}  \implies \alpha_i= \left( \frac{d_{i,i}}{c_{i,i}} \right)^{1/2}.
\end{equation}
	
Also, since $\lambda_i^2 \theta_i^2=\alpha_i^2$, we have $\lambda_i \theta_i =\alpha_i$. If $\lambda_i \theta_i =-\alpha_i$, then we will simply choose $\beta_i=-\alpha_i$ and proceed similarly. Therefore, since $\theta_1=1$, we have
	\begin{equation}\label{Th2_Eqn_3}
\begin{aligned}
& \lambda_1=\alpha_1,~\theta_1=1,~\theta_2=\frac{\alpha_2}{\lambda_2},~\theta_3=\frac{\alpha_3}{\lambda_3},\ldots,~\theta_n=\frac{\alpha_n}{\lambda_n}.
\end{aligned}
\end{equation}
	
\noindent It is evident from  (\ref{Th2_Eqn_3}) that $D_1$ and $D_2$ take the desired form. Also, based on  (\ref{Th2_Eqn_1}), we will have
	\begin{equation*}
\begin{aligned}
d_{i,j} = \alpha_i \alpha_j c_{i,j},
\end{aligned}
\end{equation*}
	for $1 \leq i, j \leq n$.
	
For the converse part, suppose that the conditions $d_{i,j} = \alpha_i \alpha_j c_{i,j}$ hold for $i,j=1,2,\ldots,n$. Now we will show that $M=\Phi(D_1,D_2,M_1)$ is an involutory matrix, where
$D_1=\Diag(\splitatcommas{\alpha_1,\lambda_2,\lambda_3,\ldots,\lambda_n})$ and $D_2=\Diag(1,\frac{\alpha_2}{\lambda_2},\frac{\alpha_3}{\lambda_3},\ldots,\frac{\alpha_n}{\lambda_n})$.

If $M=(a_{i,j})$ then we have $a_{i,j}=\lambda_i \theta_j c_{i,j}$ for $1\leq i,j \leq n$. Therefore, we have
\begin{equation*}
\begin{aligned}
(M^2)_{i,j} &= \sum_{k=1}^n a_{i,k}a_{k,j} \\
&= \sum_{k=1}^n (\lambda_i \theta_k c_{i,k})(\lambda_k \theta_j c_{k,j})
= \sum_{k=1}^n \lambda_i \theta_j c_{i,k} \lambda_k \theta_k c_{k,j} \\
& =\sum_{k=1}^n \lambda_i \theta_jc_{i,k} \alpha_k c_{k,j} ~\hspace*{2em} [\text{since}~ \lambda_i\theta_i= \alpha_i]\\
&= \sum_{k=1}^n \lambda_i\theta_jc_{i,k}\frac{d_{k,j}}{\alpha_j} ~\hspace*{2.7em} [\text{since}~d_{i,j} = \alpha_i \alpha_j c_{i,j}]\\
&=\frac{\lambda_i}{\lambda_j}\sum_{k=1}^n c_{i,k}d_{k,j} ~\hspace*{3.1em} [\text{since}~ \lambda_i\theta_i= \alpha_i]\\
&=
\begin{cases}
0 & i \neq j \\
1 & i = j \\
\end{cases} ~\hspace*{4.2em}  [\text{since} ~M_2=M_1^{-1}]
\end{aligned}
\end{equation*}
	
\noindent This demonstrates that $M^2=I$, implying that $M$ is an involutory matrix. This completes the proof.
\end{proof}

\begin{corollary}
Let $M=\Phi(D_1,D_2,M_1)$ be an involutory matrix. Then the following conditions apply to the entries of the matrices $M_1$ and $M_2$, as defined in the Theorem \ref{thm:2}:
\begin{enumerate}
\item $\frac{d_{i,i}}{c_{i,i}}\in QR(p^m)$ for  $i=1,2,3,\ldots,n$. In particular, $d_{1,1}\in QR(p^m)$.
\item $\frac{d_{i,j}}{c_{i,j}}=\frac{d_{j,i}}{c_{j,i}}$. In particular, $d_{i,1}=d_{1,i}$ for $i=1,2,\ldots,n$.
\end{enumerate}
\end{corollary}
\begin{proof}
According to (\ref{3.5}),  $\frac{d_{i,i}}{c_{i,i}}=\alpha_i^2$ for each $i=1, 2, \ldots, n$. This indicates that  $\frac{d_{i,i}}{c_{i,i}}\in QR(p^m)$ for  $i=1,2, \ldots,n$. Moreover, we conclude that $d_{1,1}\in QR(p^m)$ because $c_{1, 1}=1$. This completes the proof of part $1$. The proof of part $2$ easily follows from Theorem \ref{thm:2} and the fact that $c_{i, 1}=c_{1, i}=1$ for $i=1,2, \ldots,n$ (see (\ref{eq:1.5})).
\end{proof}

The next theorem determines the criteria on $R$ that must be met for the representative matrix $M_1$ to be an MDS matrix.

\begin{theorem}\label{thm:3}
Let $M_1$ be a representative MDS matrix of order $n$, then $R$ meets the requirements listed below:
\begin{enumerate}
\item $R$ is an MDS matrix.
\item For $i,j\in\{1,2,\ldots,n-1\}$, the modified matrix $R$ created by substituting the $i^{th}$ row or $j^{th}$ column or both of $R$ with all ones, is non-singular.
\item No entry of $R$ can be $1$.
\item A column or row of $R$ has distinct elements.
\item $R-U$ is non-singular, where $U$ is an $(n-1)\times (n-1)$ matrix with all entries $1$.
\end{enumerate}
\end{theorem}
\begin{proof}
We will demonstrate that the requirements $1$ to $5$ are true by assuming that $M_1$ is an MDS matrix. If $M_1$ is MDS, $R$ must also be MDS. So, the first condition is true.
For condition $2$, if we replace the $i^{th}$ row or $j^{th}$ column or both of $R$ by all ones, then the modified $R$ will be one  of the following forms:
\begin{enumerate}[a.]
\item Replacing  $i^{th}$ row of $R$ by all ones.
\item Replacing $j^{th}$ column of $R$ by all ones.
\item Applying both of the operations mentioned above in $a$ and $b$ simultaneously.
\end{enumerate}
In each of these instances, the modified $R$ is either a $(n-1) \times (n-1)$ sub-matrix of $M_1$, or it is reducible to some $(n-1) \times (n-1)$ sub-matrix of $M_1$. Since $M_1$ is MDS, the modified $R$ is therefore non-singular.

If any element of $R$ is $1$, then there exists a $2\times2$ sub-matrix $\begin{pmatrix}1&1\\1&1\end{pmatrix}$ of $M_1$. The determinant of this sub-matrix is zero. This is not possible because $M_1$ is MDS. Therefore, no entry of $R$ can be $1$. 

Next, if any row or column of $R$ has a repeated entry, a $2 \times 2$ sub-matrix of $M_1$ of the following form  $\begin{pmatrix}1&1\\x&x\end{pmatrix}\ \text{or}\ \begin{pmatrix}1&x\\1&x\end{pmatrix}$ must exist. The determinant of this sub-matrix is zero, which is impossible as $M_1$ is an MDS matrix. This establishes the fourth condition. At last, we demonstrate that the fifth condition is true. We transform $M_1$ into $M_1'$ by applying a series of row operations $R_i\leftarrow R_1+R_i, i=2,3,\ldots,n$, where
\begin{equation}\label{11}
 M_1^{'}=
\begin{pmatrix}
1&1&\ldots&1\\
0&&&\\
\vdots&&R-U&\\
0&&&
\end{pmatrix}.
\end{equation}
The determinant of $M_1'$ is $|R-U|$. Since $M_1$ is non-singular, so is $R-U$. Thus, condition $5$ holds and this completes the proof.
\end{proof}

The following theorem derives the sufficient conditions on $R$ that ensure that all minors of  $M_1$ of order $1$, $2$, $n-1$, and $n$ are non-zero.

\begin{theorem}\label{thm:4}
Let $M_1$ be a representative matrix of order $n$. If $R$ satisfies the conditions $1–5$ mentioned in Theorem $\ref{thm:3}$. Then, all minors of $M_1$ of order $1$, $2$, $n-1$, and $n$ are non-zero.
\end{theorem}

\begin{proof}
Let $R=(r_{i,j})$. As $R$ is MDS, we note that no element of $R$ can be $0$. Furthermore, every element in the first row and first column of $M_1$ is $1$. As a result, none of the entries of $M_1$ are zero, i.e., all of its minors of order $1$ are non-zero. Further, any $2\times 2$ sub-matrix of $M_1$ is either a sub-matrix of $R$ or one of the following forms:
$$\begin{pmatrix}
1&1\\
r_{i,j}	&r_{i,j+1}
\end{pmatrix}, \  \begin{pmatrix}
1&1\\
1&r_{i,j}\\
\end{pmatrix}, \ \begin{pmatrix}
1&r_{i,j}\\
1&r_{i+1,j}
\end{pmatrix}.$$
If it is a sub-matrix of $R$, it is non-singular because $R$ is MDS. If it is one of the three sub-matrices stated above, its determinant will be one of $r_{i,j}-r_{i,j+1}, r_{i,j}-1$, and $r_{i,j}-r_{i+1,j}$, respectively. These three determinants are not zero since no element of $R$ is $1$ and each element in each row and column of $R$ is distinct.  This indicates that none of the minors of $M_1$ of order $2$ are zero. Next, we observe that each $(n-1)\times(n-1)$ sub-matrix of $M_1$ is either $R$ or the sub-matrix obtained by doing one of the following three operations:
\begin{enumerate}[a.]
\item Removing the $i^{th}$ row from $R$ to create a $(n-2)\times(n-1)$ matrix $R'$, then adding a row of all ones on top of $R'$.
\item Removing the $j^{th}$ column from $R$ to obtain a $(n-1)\times(n-2)$ matrix $R''$, then inserting a column of all ones at the leftmost position of $R''$.
\item Applying both of the operations mentioned above in $a$ and $b$ simultaneously.
\end{enumerate}
So, if the $(n-1)\times(n-1)$ sub-matrix of $M_1$  is $R$, then it is non-singular by condition $1$ of Theorem $\ref{thm:3}$. We will now discuss the cases in $a$, $b$, and $c$. It is possible to reduce the sub-matrix of the case $a$ to a matrix by replacing the $i^{th}$ row of $R$ with all ones. Similar to the first case, the sub-matrices in the other two cases are reducible to the matrices formed by replacing either the $j^{th}$ column or both the $i^{th}$ row and the $j^{th}$ column of $R$ with all ones. Condition $2$ of Theorem $\ref{thm:3}$ guarantees that none of these sub-matrices is singular for any of these three cases. Consequently, all minors of $M_1$ of order $n-1$ are non-zero.

Finally, after applying a sequence of row operations $R_i\gets R_1+R_i, i=2,3,\ldots,n$, we  transform  $M_1$ into $M_1'$, where $M_1'$ is same as in (\ref{11}). The determinant of $M_1$ is $|R-U|$. Since $R-U$ is non-singular by condition $5$ of Theorem $\ref{thm:3}$, hence, $M_1$ is also non-singular. Thus, we have proved that all minors of $M_1$ of order $1$, $2$, $n-1$ and $n$ are non-zero. This completes the proof.
\end{proof}

\begin{corollary}\label{cor:2}
The conditions of Theorem $\ref{thm:3}$ are sufficient for a representative matrix $M_1$ of order $2$, $3$, and $4$ to be an MDS matrix.
\end{corollary}
\begin{proof}
Let $M_1$ be a representative matrix of order $n$. For $n=2$, since $M_1$ only includes minors of sizes $1$ and $2$, Theorem $\ref{thm:4}$ implies that all minors of $M_1$ are non-zero. Next, let $n$ equal $3$. Consequently, Theorem $\ref{thm:4}$ suggests once more that every minor of $M_1$ is non-zero (In this case, $M_1$ includes minors of sizes $1$, $2$, and $3$). Theorem $\ref{thm:4}$ also concludes that all minors of $M_1$ are non-zero for $n=4$ ($M_1$ has minors of sizes $1$,$2$, $3$, and $4$).  In each of these cases, $M_1$ is an MDS matrix.
\end{proof}

\begin{remark}\label{Remark_iff_condition}
It is important to note that if $n\geq 5$, the conditions outlined in Theorem $\ref{thm:3}$ are not sufficient for $M_1$ to be an MDS matrix. For instance, when $n=5$, Theorem $\ref{thm:4}$ only guarantees that all minors of $M_1$ of order $1$, $2$, $4$, and $5$ are non-zero. However, it does not address the minors of order $3$. However, the conditions outlined in Theorem $\ref{thm:3}$ serve as both necessary and sufficient criteria on $R$ for a matrix $M_1$ of order $n$, where $2\leq n\leq 4$, to be an MDS matrix.
\end{remark}

Theorem $\ref{thm:2}$ also discusses the necessary and sufficient conditions on the matrices $D_1, D_2$, and $M_1$ for $M$ to be an involutory matrix, where $M=D_1M_1D_2$. Thus, the necessary and sufficient conditions for a matrix of order $n$ to be an involutory MDS matrix, where $2\leq n\leq 4$, are obtained by combining Theorems $\ref{thm:2}$, $\ref{thm:3}$, and  Corollary $\ref{cor:2}$.

Now, we present our algorithms for generating MDS and involutory MDS matrices over $\mathbb{F}_{p^m}$ using the results of Theorems \ref{thm:1}, \ref{thm:2}, \ref{thm:3} and \ref{thm:4} (as well as Corollaries).

Let $\mathcal{R}_n$ denote the collection of all $n\times n$ matrices that satisfy the five conditions outlined for $R$ in Theorem \ref{thm:3}. Let $\alpha$ be a fixed element of $\mathbb{F}_{p^m}^*$ and $M_n(\mathbb{F}_{p^m}^*)$ represents the set of all $n\times n$ matrices over $\mathbb{F}_{p^m}^*$.  Now, let us consider the following sets of diagonal matrices:
$$ \mathcal{D}_n^{(1)}=\{\Diag(\lambda_1,\lambda_2,\ldots,\lambda_n)\mid\lambda_1,\lambda_2,\ldots,\lambda_n\in \mathbb{F}_{p^m}^*\},$$  $$\mathcal{D}_n^{(2)}=\{\Diag(1,\theta_2,\ldots,\\ \theta_n)\mid\theta_2,\theta_2,\ldots,\theta_n\in \mathbb{F}_{p^m}^*\},$$ $$\ \mathcal{D}_{\alpha,n}=\{\Diag(\alpha,\lambda_2,\ldots,\lambda_n)\mid\lambda_2,\ldots,\lambda_n\in\mathbb{F}_{p^m}^*\}.$$

We utilize Theorems \ref{thm:1} and \ref{thm:3} to formulate our first algorithm (i.e., Algorithm~\ref{Revised_Algorithm_1}) for generating all $n \times n$ MDS matrices. Specifically, it follows the steps below:

\begin{enumerate}
\item First, we generate all representative MDS matrices of order $n$.
\item From these representative matrices, we generate all $n\times n$ MDS matrices.
\end{enumerate}

We employ Theorems $\ref{thm:1}$, $\ref{thm:2}$ and $\ref{thm:3}$  to formulate our second algorithm (i.e., Algorithm~\ref{Algorithm_2}) for generating all $n\times n$ involutory MDS matrices. We specifically follow the steps below:

\begin{enumerate}
\item First, we generate all representative MDS matrices of order $n$. We observe that the word representative refers to a class of MDS matrices. Since a representative matrix of an involutory MDS matrix is not an involutory MDS matrix, it is not utilized to designate a class of involutory MDS matrices.
\item We then determine whether or not the representative MDS matrix satisfies the involutory conditions of Theorem \ref{thm:2}. The best part of this strategy is that the representative itself can be tested. It is not necessary to obtain the exact matrix.
\item Finally, we established the requirements for $D_1$ and $D_2$ as in Theorem \ref{thm:2}.
\end{enumerate}

\begin{algorithm}[htpb!]
\caption{Printing all $n\times n$  MDS matrices over $\mathbb{F}_{p^m}$.}
\label{Revised_Algorithm_1}
{\bf Input:} Finite Field $\mathbb{F}_{p^m}$ and positive integer $n \geq 2$.\\
{\bf Output:} All MDS matrices of order $n\times n$.
\begin{algorithmic}[1]
\For{$R\in M_{n-1}(\mathbb{F}_{p^m}^*)$}
\If{$R \notin \mathcal{R}_{n-1}$}
\State { \bf continue}
\Else
\State $M_1=\begin{pmatrix} 1&\ldots&1\\ \vdots&R&\\ 1&& \end{pmatrix}$
\If{$n>4$ and $M_1$ is not MDS} \label{restriction_n_algo1}
\State {\bf continue}
\Else
\For{$D_1\in\mathcal{D}_n^{(1)}$ and $D_2\in\mathcal{D}_n^{(2)}$}
\State Print matrix $D_1M_1D_2$
\EndFor
\EndIf
\EndIf
\EndFor
\end{algorithmic}
\end{algorithm}
\setcounter{algorithm}{1}

\begin{remark}
The line \ref{restriction_n_algo1} of Algorithm~\ref{Revised_Algorithm_1} employs a filter to eliminate non-MDS representative matrices when $n>4$. According to Remark~\ref{Remark_iff_condition}, for $n\leq 4$, if $R\in \mathcal{R}_{n-1}$ then $M_1$ will be a representative MDS matrix. Hence, for $n\leq 4$, Line \ref{restriction_n_algo1} avoids filtering and provides a direct computational advantage over other cases.
\end{remark}

In the following example, we consider an MDS matrix from the literature ~\cite[Example~4]{kcz} and demonstrate its corresponding triplets of matrices $(D_1, D_2, M_1)$, which are utilized to produce the MDS matrix.

\begin{example}\label{exm:1}
Let $\beta$ be a root of the primitive polynomial $x^4+x+1$, which generates $\mathbb{F}_{2^4}$. Consider the following $4 \times 4$ MDS matrix $M$ over $\mathbb{F}_{2^4}$.

\[M=(a_{i,j})=\begin{pmatrix}
\beta^3 + \beta^2+ 1   & \beta^2 +\beta + 1 & \beta^3 + \beta & \beta+1 \\
\beta^2 +\beta + 1  & \beta^3 + \beta^2+ 1  & \beta + 1  & \beta^3 + \beta \\
\beta^3 + \beta & \beta + 1  & \beta^3+\beta^2 + 1 & \beta^2+\beta+1\\
\beta + 1  & \beta^3 + \beta   & \beta^2 +\beta+1 & \beta^3+\beta^2 + 1
\end{pmatrix}.\]
Using Theorem \ref{thm:1}, $M$ can be written as $$M=D_1M_1D_2,$$ where  $D_1=\Diag(\lambda_1, \lambda_2, \lambda_3, \lambda_4)$ and $D_2=\Diag(1, \theta_1, \theta_2, \theta_3)$ are unique non-singular diagonal matrices, and $M_1$ is a representative MDS matrix.

The representative MDS matrix $M_1$ of the class to which the MDS matrix $M$ belongs is represented as follows:
\[M_1=\begin{pmatrix}
1 & 1 & 1 & 1 \\
1 & \beta^3+\beta^2   & \beta^3+\beta^2+1   & \beta^2 + 1 \\
1 & \beta^3+\beta^2+1 & \beta^2+1           & \beta^2+\beta+1\\
1 & \beta^2+1         & \beta^2+\beta+1     & \beta^3
\end{pmatrix}.\]
We can find unique values of $D_1=\Diag(\lambda_1, \lambda_2, \lambda_3, \lambda_4)$ and $D_2=\Diag(1, \theta_2, \theta_3,  \theta_4)$ from Theorem \ref{thm:1}.
\begin{eqnarray*}	
&& ~\lambda_1=\beta^3+\beta^2+1, ~\lambda_2=\beta^2+\beta+1, ~\lambda_3=\beta^3 + \beta, ~\lambda_4=\beta + 1\\
&& ~\theta_2=\beta^3+\beta^2+\beta+1,~\theta_3=\beta^3+\beta^2+\beta,~\theta_4=\beta^3 + \beta^2.
\end{eqnarray*}	
The MDS matrix $M$ can easily be obtained by applying $D_1=\Diag(~\beta^3+\beta^2+1,~\beta^2+\beta+1,~\beta^3 + \beta,~\beta + 1)$ and $D_2=\Diag(~1,~\beta^3+\beta^2+\beta+1,~\beta^3+\beta^2+\beta,~\beta^3 + \beta^2)$ to representative MDS matrix $M_1$.
\end{example}

 \begin{algorithm}[htpb!]
\caption{Printing all $n\times n$ involutory MDS matrices over $\mathbb{F}_{p^m}$.}\label{Algorithm_2}
{\bf Input:} Finite Field $\mathbb{F}_{p^m}$ and positive integer $n \geq 2$.\\
{\bf Output:} All involutory MDS matrices of order $n\times n$.
\begin{algorithmic}[1]
\For{$R\in M_{n-1}(\mathbb{F}_{p^m}^*)$}
\If{$R \notin \mathcal{R}_{n-1}$} { \bf continue}
\Else
\State $M_1=\begin{pmatrix} 1&\ldots&1\\ \vdots&R&\\ 1&& \end{pmatrix}$
\If{$n>4$ and $M_1$ is not MDS} { \bf continue}
\Else
\State $M_2=(d_{i,j})\gets  M_1^{-1}$
\State{$flag\gets 0$}
\For{$i\gets 1$  to $n$ }
\If{ $\frac{d_{i,i}}{c_{i,i}}\notin QR(p^m)$}{ $flag\gets 1$}
\State{ \bf goto~\ref{A}}
\EndIf
\State $\alpha_i\gets \left(\frac{d_{i,i}}{c_{i,i}}\right)^{1/2}$
\EndFor
\If{$flag=1$}{ \bf continue}\label{A}
\EndIf
\State $flag\gets 0$
\For{$i\gets 1$ to $n$}
\For{$j\gets i+1$ to $n$}
\If{$d_{i,j}c_{j,i}\neq d_{j,i}c_{i,j}$ or $d_{i,j}\neq\alpha_i\alpha_jc_{i,j}$}{ $flag\gets 1$}
\State{ \bf goto~\ref{A'}}
\EndIf
\EndFor
\EndFor
\If{$flag=1$}{ \bf continue}\label{A'}
\EndIf
\For{$D_1\gets \Diag(\alpha_1,\lambda_2,\ldots,\lambda_n)\in\mathcal{D}_{\alpha_1,n}$}
\State $D_2\gets \Diag(1,\frac{\alpha_2}{\lambda_2},\ldots,\frac{\alpha_n}{\lambda_n})$
\State Print matrix $D_1M_1D_2$.
\EndFor
\EndIf
\EndIf
\EndFor
\end{algorithmic}
\end{algorithm}

\begin{remark}
The steps $8$,~$10$,~$11$,~$12$, ~$15$ and $16$ in Algorithm~\ref{Algorithm_2} are redundant when the underlying field is of characteristic $2$. It is evident from the fact that any non-zero element in a finite field of characteristic $2$ is a quadratic residue.
\end{remark}

\begin{remark}
It is worth mentioning that in~\cite{Yang2021}, the authors show that for finding all $n\times n$ involutory MDS matrices over $\mathbb{F}_{2^m}$, the search space can be reduced to $2^{(mn^2)/2}$. Whereas, in our approach, by finding the representative MDS matrices over $\mathbb{F}_{2^m}$, the search space is $2^{m(n-1)^2}$, which exceeds the former when $n\geq 4$. However, it is important to emphasize that our approach is more general, as it not only focuses on involutory MDS matrices but also considers the case of finding all $n\times n$ MDS matrices over $\mathbb{F}_{p^m}$.
\end{remark}

In the following example, we consider an involutory MDS matrix from the literature  ~\cite[Example~7]{kcz} and demonstrate its corresponding triplets of matrices $(D_1, D_2, M_1)$, which are utilized to produce the involutory MDS matrix.

\begin{example}\label{exm:2}
Let $\beta$ be a root of the primitive polynomial $x^4+x+1$, which generates $\mathbb{F}_{2^4}$. Consider the following $4 \times 4$ involutory MDS matrix $M$ over $\mathbb{F}_{2^4}$.
	
\[M=(a_{i,j})=\begin{pmatrix}
\beta^3 + \beta   & \beta^3 + \beta^2 & \beta^2 +\beta & 1 \\
\beta^3 + \beta^2 & \beta^3 + \beta   &  1      & \beta^2 +\beta\\
\beta^2 +\beta & 1 & \beta^3 + \beta   & \beta^3 + \beta^2 \\
 1      & \beta^2 +\beta &\beta^3 + \beta^2 & \beta^3 + \beta
\end{pmatrix}.\]

Since $M$ is involutory, using Theorems \ref{thm:1} and \ref{thm:2}, $M$ can be written as $$M=D_1M_1D_2,$$ where $D_1=\Diag(\alpha_1, \lambda_2, \lambda_3, \lambda_4)$ and $D_2=\Diag(1, \frac{\alpha_2}{\lambda_2}, \frac{\alpha_3}{\lambda_3}, \frac{\alpha_4}{\lambda_4})$ are unique non-singular diagonal matrices, and $M_1$ is a representative MDS matrix.

The representative MDS matrix $M_1$ of the class to which the MDS matrix $M$ belongs is represented as follows:
\[M_1=(c_{i,j})=\begin{pmatrix}
1 & 1 & 1 & 1 \\
1 &\beta^3+\beta^2  &\beta^3+\beta^2+1  &\beta^2+1 \\
1 &\beta^3+\beta^2+1&\beta^2+1          &\beta^2 +\beta +1  \\
1 &\beta^2+1        &\beta^2+\beta+1    &\beta^3
\end{pmatrix}.\]

The inverse of representative MDS matrix $M_1$ is given as:
\[M_1^{-1}=(d_{i,j})=\begin{pmatrix}
\beta^3&\beta^3+\beta^2+\beta+1&\beta^2+\beta+1&1 \\
\beta^3+\beta^2+\beta+1&\beta^3+\beta^2+\beta+1&\beta^2&\beta^2\\
\beta^2+\beta+1 &\beta^2&\beta^2+\beta+1 &\beta^2\\
1&\beta^2&\beta^2&1
\end{pmatrix}.\]

\noindent We can find unique values of $D_1=\Diag(\alpha_1, \lambda_2, \lambda_3, \lambda_4)$ and $D_2=\Diag(1, \frac{\alpha_2}{\lambda_2}, \frac{\alpha_3}{\lambda_3}, \frac{\alpha_4}{\lambda_4})$ from Theorems \ref{thm:1} and \ref{thm:2}.
\begin{eqnarray*}
&&~\lambda_2=\beta^3+\beta^2,~\lambda_3=\beta^2+\beta,~\lambda_4=1\\
&&\alpha_i=\left(\frac{d_{i,i}}{c_{i,i}}\right)^{1/2}, i=1,2,3,4\\		
&&~\alpha_1= \beta^3+\beta, ~\alpha_2=\beta^3, ~\alpha_3=\beta, \alpha_4=\beta^3+\beta^2.\\
\end{eqnarray*}	
The involutory MDS matrix $M$ can easily be obtained by applying $D_1=\Diag(\beta^3+\beta,~\beta^3+\beta^2,~\beta^2+\beta,~1)$ and $D_2=\Diag(1,~\beta^3+\beta^2+\beta+1,~\beta^3+\beta^2+\beta,~\beta^3 + \beta^2)$ to representative MDS matrix $M_1$.
\end{example}

\section{Counting of all $3\times 3$ MDS and involutory MDS matrices over $\mathbb{F}_{2^m}$}\label{Sec:Counting_3_MDS}
In \cite{gmt}, the authors provide an explicit formula for enumerating all $3\times3$ involutory MDS matrices over $\mathbb{F}_{2^m}$. However, a precise formula for calculating the number of all $3\times3$ MDS matrices over $\mathbb{F}_{2^m}$ remains absent in existing literature. In this section, we take advantage of our proposed method to determine the conditions under which the representative matrix becomes an MDS matrix. We explore the possible choices of representative matrices $M_1$ based on these conditions. Then, by considering the number of appropriate choices for $D_1$ and $D_2$, we can easily find the total number of MDS matrices $\Phi(D_1,D_2,M_1)$. Lemma~\ref{Lemma_count_3_MDS} outlines the conditions for the representative matrix $M_1$ to be MDS, and Theorem~\ref{Th_count_3_MDS} provides the formula for enumerating all $3\times 3$ MDS matrices over $\mathbb{F}_{2^m}$.

\begin{lemma}\label{Lemma_count_3_MDS}
The representative matrix $M_1$ of order $3\times3$ over $\mathbb{F}_{2^m}$ as follows:
\begin{equation*}
\begin{aligned}
M_1 & =
\begin{bmatrix}
1 & 1 & 1 \\
1 & a & b \\
1 & c & d
\end{bmatrix}
\end{aligned}
\end{equation*}
is MDS if and only if the tuple $(a, b, c, d) \in \mathbb{F}_{2^m}^4$ satisfies the  conditions: (i) $a, b, c, d \in \mathbb{F}_{2^m}\setminus \set{0,1}$, (ii) $b \neq a$,  $c \neq a$,  $d \neq b$,  $d \neq c$,  $d \neq a^{-1}b c$, and (iii) $d+1 \neq (a+1)^{-1}(b+1) (c+1)$.
\end{lemma}

\begin{proof}
We can infer from Corollary \ref{cor:2} that $M_1$ will be an MDS matrix if and only if the conditions mentioned in Theorem \ref{thm:3} hold. Thus, $M_1$ is an MDS matrix if and only if the following conditions hold for $R=\begin{pmatrix}
a &b\\c &d
\end{pmatrix}$:

\begin{enumerate}
\item $R$ is MDS. This is true if and only if $a,b,c,d$ are not equal to $0$ and $ad-bc\neq0$ i.e., $d \neq a^{-1}b c$.
\item Given $i,j\in\{1,2\}$, the modified matrix $R$ obtained by replacing $i^{th}$ row or $j^{th}$ column or both of $R$ by all ones is non-singular. It means $b\neq a$, $c\neq a$, $d\neq c$, $d\neq b$, and $a,b,c,d \neq 1$.
\item  No entry of $R$ can be $1$. This is already covered in condition $2$.
\item A column or row of $R$ has distinct elements. This is also covered in condition $2$.
\item $R-U$ is non-singular, where $U$ is a $2\times 2$ matrix with all entries $1$. This is equivalent to the statement $|R-U|=|M_1|\neq 0$, i.e., $$d+1 \neq (a+1)^{-1}(b+1) (c+1).$$
\end{enumerate}
This completes the proof.
\end{proof}

\begin{theorem}\label{Th_count_3_MDS}
The count of all representative MDS matrices of order  $3\times 3$ over the finite field $\mathbb{F}_{2^m}$ is given by $(2^m-2)(2^m-3)(2^{2m}-9\cdot 2^m+21)$.
\end{theorem}

\begin{proof}
According to Lemma~\ref{Lemma_count_3_MDS}, the number of all representative MDS matrices of order  $3\times 3$ over $\mathbb{F}_{2^m}$ is equal to the cardinality of the set $S$, defined as:
	
\begin{center}
\begin{small}
$S=\{\splitatcommas{(a, b, c, d)\in (\mathbb{F}_{2^m}\setminus \set{0,1})^4: ~b \neq a, ~c \neq a, ~d \neq b, ~d \neq c, ~d \neq a^{-1}b c, ~d+1 \neq (a+1)^{-1}(b+1) (c+1)}\}$.
\end{small}
\end{center}
	
Since the condition $d+1 \neq (a+1)^{-1}(b+1) (c+1)$ implies that the representative matrix is non-singular, to find the cardinality of $S$ we will find all non-singular matrices of the form
\begin{equation*}
\begin{bmatrix}
1 & 1 & 1 \\
1 & a & b \\
1 & c & d
\end{bmatrix},
\end{equation*}
	where ~$a, b, c, d \in \mathbb{F}_{2^m}\setminus \set{0,1}$, and ~$b \neq a$, ~$c \neq a$, ~$d \neq b$, ~$d \neq c$, ~$d \neq a^{-1}b c$.
	
It is easy to check that,
\begin{equation*}
\begin{bmatrix}
1 & 1 & 1 \\
1 & a & b \\
1 & c & d
\end{bmatrix}~\text{is non-singular}~\iff~
\begin{bmatrix}
a+1 & b+1 \\
c+1 & d+1
\end{bmatrix}~\text{is non-singular}.
\end{equation*}
	
Therefore, the problem of finding the cardinality of $S$ is now reduced to the problem of finding all non-singular matrices of the form
	\begin{equation}\label{Eqn_2_Rep_MDS}
\begin{aligned}
A & =
\begin{bmatrix}
a+1 & b+1 \\
c+1 & d+1
\end{bmatrix},
\end{aligned}
\end{equation}
where ~$a, b, c, d \in \mathbb{F}_{2^m}\setminus \set{0,1}$, and ~$b \neq a$, ~$c \neq a$, ~$d \neq b$, ~$d \neq c$, ~$d \neq a^{-1}b c$.

Now, since $a \in \mathbb{F}_{2^m}\setminus \set{0,1}$, we have $2^m-2$ many choices for $a$. Also, since $b \in \mathbb{F}_{2^m}\setminus \set{0,1}$ and $b \neq a$, we have $2^m-3$ many choices for $b$. Therefore, for the first row of $A$, there are $(2^m-2)(2^m-3)$ many choices.
	
Since $c \in \mathbb{F}_{2^m}\setminus \set{0,1}$ and $c\neq a$, for $c$, we have $2^m-3$ many choices over $\mathbb{F}_{2^m}$. Also, we can see that $d \not \in \set{0,1,b,c,a^{-1}bc}$. Moreover, to ensure the non-singularity of $A$, we need to exclude the linear combinations of the first row from the total choices available for the second row. Specifically, we need to exclude the cases when
	\begin{center}
$(c+1, d+1) = k(a+1, b+1)$ for some $k \in \mathbb{F}_{2^m}$.
\end{center}
	
However, specific linear combinations have already been omitted based on the choices of $c$ and $d$. For example, cases where $k=0, 1$ have already been excluded because of the conditions $c,d \in \mathbb{F}_{2^m}\setminus \set{0,1}$ and $c \neq a$ and $d \neq b$.
	
Also, when $k=(a+1)^{-1}$, we have $c+1=1\implies c=0$, which have already been excluded from the choices of $c$. Similarly, the choice for $k=(b+1)^{-1}$ have already been excluded.
	
Now, $c+1=k(a+1)$, we have $c=k(a+1)+1$. Similarly, $d=k(b+1)+1$. Thus,
	$$d =c \implies k(b+1)+1 = k(a+1)+1 \implies b=a,$$ which contradicts the fact $b\neq a$.
	
Moreover,
\begin{equation*}
\begin{aligned}
& d = a^{-1}bc\\
\implies & k(b+1)+1 = a^{-1}b [k(a+1)+1]\\
\implies & k(b+1)+1 = k (b+ a^{-1}b)+a^{-1}b\\
\implies & k(a^{-1}b+1) = a^{-1}b+1\\
\implies & k=1.
\end{aligned}
\end{equation*}
	
Thus, the values of $k \in \set{0,1,(a+1)^{-1},(b+1)^{-1}}$ have already been omitted due to the choices of $c$ and $d$. Therefore, to ensure the non-singularity of $A$, we need to exclude $2^m-4$ many linear combinations from the total choices available for the second row.
	

\begin{figure}[h]
\centering
\includegraphics[width=0.8\linewidth]{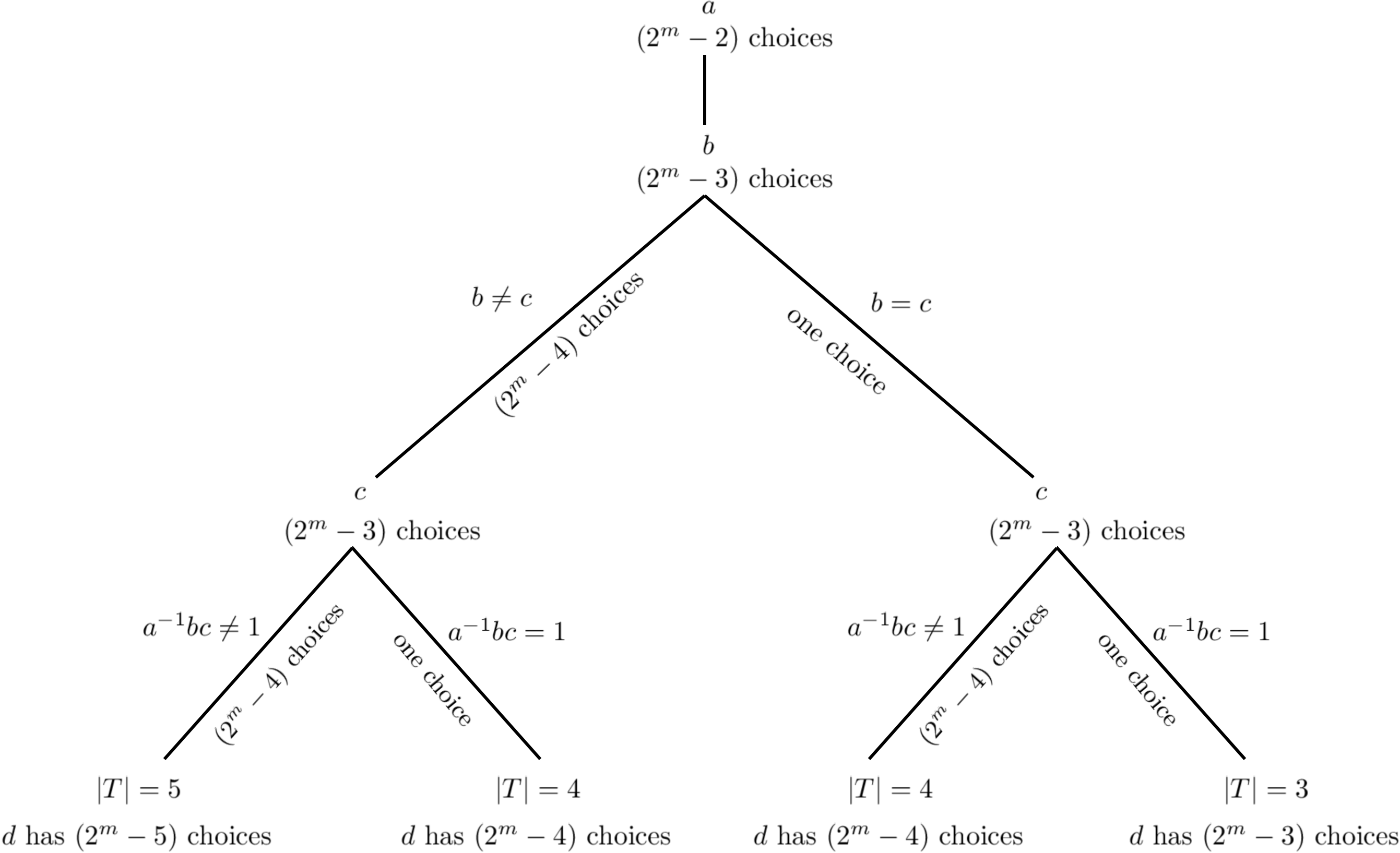}
\caption{A figure illustrating the cases for determining the number of choices for $d$.}
\label{fig:proof_counting}
\end{figure}

Now, we can see that $d \not \in T=\set{0,1,b,c,a^{-1}bc}$. However, it is important to note that each element of $T$ may not be distinct. When $c=ab^{-1}$, we find that $a^{-1}bc=1$. Additionally, $b=c$ may hold. Now,
	\begin{equation*}
\begin{aligned}
& a^{-1}bc=b \implies a=c~~\text{and}~~ a^{-1}bc=c \implies a=b,
\end{aligned}
\end{equation*}
	which leads to contradictions. Moreover, since $a,b,c \in \mathbb{F}_{2^m}^{*}$, we have $a^{-1}bc \neq 0$.

We will now analyze two cases to determine the cardinality of the set $T$.
	
\begin{enumerate}[I.]
\itemsep1em
		
\item \textbf{Case 1:} $b \neq c$. We can further divide this case into two subcases:
		
\begin{enumerate}[(i)]
\item \textit{Subcase 1.1:} $a^{-1}bc \neq 1$. In this scenario, the cardinality of set $T$ is $|T|=5$. Therefore, for $d$, there are $2^m-5$ choices over $\mathbb{F}_{2^m}$.
			
\item \textit{Subcase 1.2:} $a^{-1}bc=1$. In this instance, the cardinality of set $T$ is $|T|=4$. Consequently, for $d$, there are $2^m-4$ choices over $\mathbb{F}_{2^m}$.
\end{enumerate}	
		
\item \textbf{Case 2:} $b=c$. We can further divide this case into two subcases:
		
\begin{enumerate}[(i)]
\item \textit{Subcase 2.1:} $a^{-1}bc \neq 1$. In this scenario, the cardinality of set $T$ is $|T|=4$. Therefore, for $d$, there are $2^m-4$ choices over $\mathbb{F}_{2^m}$.
			
\item \textit{Subcase 2.2:} $a^{-1}bc=1$. In this instance, the cardinality of set $T$ is $|T|=3$. Consequently, for $d$, there are $2^m-3$ choices over $\mathbb{F}_{2^m}$.
\end{enumerate}
\end{enumerate}	
	
Now, depending on the four cases, we can determine the total number of matrices of the form $A$ (see (\ref{Eqn_2_Rep_MDS})). Figure~\ref{fig:proof_counting} illustrates the cases for determining the choices of $d$ in each case. Also, as mentioned earlier, to find the number of all non-singular matrices in the form of $A$, we must exclude the $2^m-4$ linear combinations of the first row from the total choices available for the second row.
	
Therefore, the total number of all representative MDS matrices of order $3$ over $\mathbb{F}_{2^m}$ is given by
	\begin{align*}
& (2^m-2)[(2^m-4) \{(2^m-4)(2^m-5)+1\cdot (2^m-4) - (2^m-4)\}\\
&\hspace*{2cm} + 1\cdot \{(2^m-4)(2^m-4)+1\cdot (2^m-3)  - (2^m-4)\}]\\
&= (2^m-2)[(2^m-4)(2^m-4)(2^m-5)+(2^m-4)(2^m-5)+(2^m-3)]\\
&= (2^m-2)[(2^m-4)(2^m-5)\{(2^m-4)+1 \}+(2^m-3)]\\
&= (2^m-2)[(2^m-3)\{(2^m-4)(2^m-5)+1 \}]\\
&= (2^m-2)(2^m-3)(2^{2m}-9\cdot 2^m+21).
\end{align*}
This completes the proof.
\end{proof}

From a single representative MDS matrix $M_1$ of order $n$, one can generate the MDS matrices $\Phi(D_1,D_2,M_1)$, where $D_1$ and $D_2$ are non-singular diagonal matrices of order $n$ with $1$ as the first entry of $D_2$. Thus, given a representative MDS matrix $M_1$, we can generate $(2^m-1)^{2n-1}$ MDS matrices. Therefore, by applying Theorem~\ref{Th_count_3_MDS}, we can derive the explicit formula for counting all $3\times 3$ MDS matrices over the finite field $\mathbb{F}_{2^m}$.

\begin{corollary} \label{cor:MDS_Count}
The count of all $3\times 3$ MDS matrices over the finite field $\mathbb{F}_{2^m}$ is given by $(2^m-1)^5(2^m-2)(2^m-3)(2^{2m}-9\cdot 2^m+21)$.
\end{corollary}

In the following lemma, we mention the explicit formula obtained in~\cite{gmt} for enumerating all $3\times 3$ involutory MDS matrices over $\mathbb{F}_{2^m}$.

\begin{lemma}\cite{gmt}\label{lemma:inv_3_MDS_count}
The number of all $3 \times 3$ involutory MDS matrices over $\mathbb{F}_{2^m}$ is given by $(2^m-1)^2 (2^m-2)(2^m-4)$.
\end{lemma}

Thus, form Corollary~\ref{cor:MDS_Count} and Lemma~\ref{lemma:inv_3_MDS_count}, one can determine the explicit formula for enumerating all $3\times 3$ non-involutory MDS matrices over $\mathbb{F}_{2^m}$.

\begin{corollary}
The number of all $3 \times 3$ non-involutory MDS matrices over $\mathbb{F}_{2^m}$ is given by $(2^{m} - 1)^2(2^{m} - 2)(2^{6m} - 15\cdot 2^{5m} + 87\cdot 2^{4m} - 244\cdot 2^{3m} + 345\cdot 2^{2m} - 238\cdot 2^{m} + 67)$.
\end{corollary}

\section{Counting of all $4\times 4$ MDS and involutory MDS matrices over $\mathbb{F}_{2^m}$ }\label{Sec:Counting_4_MDS}
In this section, we enumerate all $4\times 4$ MDS and involutory MDS matrices over $\mathbb{F}_{2^m}$ for some specific values of $m$. Using our proposed method, we enumerate all $4\times 4$ MDS matrices by counting all representative MDS matrices. By the MDS conjecture, there are no $4\times 4$ MDS matrices over $\mathbb{F}_{2^2}$. Therefore, here we search for all representative MDS matrices of order $4$ over $\mathbb{F}_{2^m}$ where $m\geq 3$.

To search for all representative MDS matrices of order $4$ over $\mathbb{F}_{2^m}$, we use Theorem \ref{thm:1} to define the matrix form $M_1$ as follows:
\[M_1=\begin{pmatrix}
1 &1 &1 &1\\
1&&&\\
1&&R&\\
1&&&\\
\end{pmatrix},\]
where $R$ is a $3\times 3$ matrix. As per Theorem \ref{thm:3}, if $M_1$ is MDS, then $R$ must be MDS. We again apply Theorem \ref{thm:1} on $R$ to write it as $$R=D_1M_1'D_2, \ \text{where}\ M_1^{'}=\begin{pmatrix}
1 &  1& 1\\
1 & a& b\\
1 & c &d \\
\end{pmatrix} \text{is a representative matrix of order $3$}.$$
Here $a, b, c, d \in \mathbb{F}_{2^m}^*$, $D_1=\Diag(\lambda_1,\lambda_2,\lambda_3)$ and $D_2=\Diag(1,\theta_2,\theta_3)$ for $\lambda_1, \lambda_2, \lambda_3,\linebreak \theta_2, \theta_3\in \mathbb{F}_{2^m}^*$. This means that
$$R
=\begin{pmatrix}
\lambda_1&   \lambda_1\theta_2&   \lambda_1\theta_3\\
\lambda_2&   a\lambda_2\theta_2&  b\lambda_2\theta_3\\
\lambda_3&   c\lambda_3\theta_2&  d\lambda_3\theta_3\\
\end{pmatrix}.$$
Since the order of $M_1$ is $4$, from Corollary \ref{cor:2}, the conditions of Theorem \ref{thm:3} are sufficient for $M_1$ to be an MDS matrix. Precisely,  $M_1$ is an MDS matrix  if the following conditions hold for $R$:
\begin{enumerate}
\item
$R$ is MDS. As discussed in Lemma \ref{Lemma_count_3_MDS}, we know that $M_1^{'}$ and $R$ are MDS if
$$
a,b,c,d \neq 0,1,
~b\neq a, ~c\neq a, ~d\neq b, ~d\neq c,
~d\neq a^{-1}bc, $$
and $$~d+1 \neq (a+1)^{-1}(b+1) (c+1).$$
\item Given $i,j\in\{1,2,3\}$, the modified matrix $R$ obtained by replacing $i^{th}$ row or $j^{th}$ column or both of $R$ by all ones is non-singular. This condition is equivalent to checking that all $3\times 3$ minors of $M_1$ are non-zero except $R$, whose non-singularity is already taken care of in condition 1. Thus, we arrive at the $15$ inequations placed in Appendix \hyperlink{ap:A}{A}. \label{cond:Counting_4_MDS}
	
\item  No entry of $R$ can be $1$, which is possible if
\begin{eqnarray*}
&&\lambda_1,\lambda_2, \lambda_3 \neq 1, \lambda_1\theta_2 + 1  \neq 0,   \lambda_1\theta_3 + 1  \neq 0,  a\lambda_2\theta_2 + 1  \neq 0\\
&&b\lambda_2\theta_3 + 1  \neq 0, 	c\lambda_3\theta_2 + 1  \neq 0,  d\lambda_3\theta_3 + 1  \neq 0.
\end{eqnarray*}
	
\item Elements in a row or a column of $R$ are distinct, if
\begin{eqnarray*}
&&\theta_2,\theta_3 \neq 1, \theta_2\neq \theta_3, \lambda_1\neq \lambda_2, \lambda_2\neq \lambda_3, \lambda_3\neq \lambda_1 \\
&&b\theta_3 + 1  \neq 0, a\theta_2 + b\theta_3  \neq 0,
c\theta_2 + 1  \neq 0\\
&&d\theta_3 + 1  \neq 0, c\theta_2 + d\theta_3  \neq 0, a\theta_2 + 1  \neq 0\\
&&a\lambda_2 + \lambda_1  \neq 0, c\lambda_3 + \lambda_1  \neq 0, a\lambda_2 + c\lambda_3 \neq 0\\
&&b\lambda_2 + \lambda_1  \neq 0, d\lambda_3 + \lambda_1  \neq 0, b\lambda_2 + d\lambda_3 \neq 0.
\end{eqnarray*}

\item $R-U$ is non-singular, where $U$ is a $3\times 3$ matrix with all entries $1$. This is equivalent to the statement that  $|R-U|=|M_1|\neq 0$.
\end{enumerate}

At this stage, we are prepared to provide an exact count of all $4\times 4$ MDS and involutory MDS matrices over $\mathbb{F}_{2^3}$ and $\mathbb{F}_{2^4}$, respectively. According to the discussion in Section \ref{Sec:Counting_3_MDS}, there are $6554730$ and $18381431250$ MDS matrices of order $3$ over $\mathbb{F}_{2^3}$ and $\mathbb{F}_{2^4}$, respectively. Applying the conditions of Theorem \ref{thm:3} to these $3\times3$ MDS matrices, we determine that there are $720$ and $464227344$ representative MDS matrices of order $4$ over $\mathbb{F}_{2^3}$ and $\mathbb{F}_{2^4}$, respectively. Additionally, the choices for $D_1$ and $D_2$ over $\mathbb{F}_{2^m}$ are $(2^m-1)^7$. Thus, there are $7^7\times 720$ and $15^7\times 464227344$ MDS matrices of order $4$ over $\mathbb{F}_{2^3}$ and $\mathbb{F}_{2^4}$, respectively.

In the following step, we count the number of all $4\times 4$ involutory MDS matrices over $\mathbb{F}_{2^3}$ and $\mathbb{F}_{2^4}$, respectively. Involutory matrices are characterized by the Theorem \ref{thm:2}. There are $48$ and $71856$ representative MDS matrices of order $4$ over $\mathbb{F}_{2^3}$ and $\mathbb{F}_{2^4}$, respectively, by meeting the requirements of Theorem \ref{thm:2}. The total number of diagonal matrices $D_1$ and $D_2$ that meet the requirements of Theorem \ref{thm:2} must also be determined for these representative MDS matrices. According to Theorem \ref{thm:2}, the choices for $D_1$ and $D_2$ over $\mathbb{F}_{2^m}$ are $(2^m-1)^3$. Thus, there are $7^3\times 48$ and $15^3\times 71856$ involutory matrices of order $4$ over $\mathbb{F}_{2^3}$ and $\mathbb{F}_{2^4}$, respectively. Table \ref{Table:1} and Table \ref{Table:2} give the findings of Section~\ref{Sec:Counting_4_MDS}.

\begin{center}
\captionof{table}{Count  of all $4\times 4$ MDS matrices over $\mathbb{F}_{2^m}$.} \label{Table:1}
\begin{tabular}{|c|c|c|c|}\hline
Finite Field & Number of representatives & Choices for $D_1,D_2$ & Total Count \\ \hline
$\mathbb{F}_{2^2}$ & $0$      & -- &$0$ \\
$\mathbb{F}_{2^3}$ & $720$    & $7^7$  &$7^7\times 720$\\
$\mathbb{F}_{2^4}$ & $464227344$ & $15^7$ &$15^7\times 464227344$\\ \hline
\end{tabular}
\vspace{2mm}\\

\captionof{table}{Count  of all $4\times 4$ involutory MDS matrices over $\mathbb{F}_{2^m}$.} \label{Table:2}
\vspace{3mm}
\begin{tabular}{|c|c|c|c|}\hline
Finite Field & Number of representatives& Choices for $D_1,D_2$ & Total Count \\ \hline
$\mathbb{F}_{2^2}$ & $0$        & --  &$0$ \\
$\mathbb{F}_{2^3}$ & $48$       & $7^3$  &$7^3\times48$\\
$\mathbb{F}_{2^4}$ & $71856$    & $15^3$ &$15^3\times71856$\\
\hline
\end{tabular}
\vspace{2mm}\\
\end{center}

\noindent It is important to note that we are not able to provide the count of all $4\times 4$ MDS and involutory MDS matrices over $\mathbb{F}_{2^m}$ for $m\geq 5$ due to the vast search space involved. Also, in a recent work~\cite{Samanta2023}, the authors have provided the explicit formula for counting MDS and involutory MDS matrices over $\mathbb{F}_{2^m}$ for a particular type of $4\times 4$ matrices. However, an explicit count for all $4\times 4$ MDS and involutory MDS matrices is still lacking. Thus, a potential avenue for future work is to provide an explicit formula for counting the MDS and involutory MDS matrices of order $n\geq 4$.

\section{Conclusion}\label{Sec:Conclusion}
Our work addresses the challenge of efficiently generating all $n \times n$ MDS and involutory MDS matrices over $\mathbb{F}_{p^m}$. For this, we propose two algorithms designed for generating $n \times n$ MDS and involutory MDS matrices, respectively, from the representative MDS matrices. Our approach offers the advantage of reducing the search space by focusing on $(n-1) \times (n-1)$ MDS matrices when constructing the representative MDS matrices of order $n$, as opposed to exhaustively searching for $n \times n$ MDS matrices. Additionally, our study is the first to provide an explicit formula for enumerating all $3 \times 3$ MDS matrices over a finite field of characteristic $2$. We also present necessary and sufficient conditions for generating all $3 \times 3$ and $4 \times 4$ MDS and involutory MDS matrices. This work may be expanded for future work to offer necessary and sufficient conditions for the generation of all $n \times n$ MDS and involutory MDS matrices over $\mathbb{F}_{p^m}$, where $n\geq 5$. Furthermore, an explicit formula for counting all MDS and involutory MDS matrices of order $n\geq 4$ is still lacking. Thus, a potential avenue for future research is to provide an explicit formula for the counting of MDS and involutory MDS matrices of order $4$ or higher.






\section*{Acknowledgments}
We would like to express our sincere gratitude to the anonymous reviewers for their detailed and insightful feedback, which significantly contributed to the improvement of this article.

\begin{appendices} \label{ap:A}
\section{The inequalities arising out of Condition \ref{cond:Counting_4_MDS} in Section \ref{Sec:Counting_4_MDS}}

\begin{eqnarray*}
&&a\lambda_1\lambda_2\theta_2 + a\lambda_2\theta_2 + \lambda_1\lambda_2\theta_2 + \lambda_1\theta_2 + \lambda_1 + \lambda_2 \neq 0\hspace{35mm}\\
&&b\lambda_1\lambda_2\theta_3 + b\lambda_2\theta_3 + \lambda_1\lambda_2\theta_3 + \lambda_1\theta_3 + \lambda_1 + \lambda_2 \neq 0\hspace{35mm}\\
&&a\lambda_1\lambda_2\theta_2\theta_3 + b\lambda_1\lambda_2\theta_2\theta_3 + a\lambda_2\theta_2 + b\lambda_2\theta_3 + \lambda_1\theta_2 + \lambda_1\theta_3 \neq 0\hspace{19mm}\\
&&a\lambda_1\lambda_2\theta_2\theta_3 + b\lambda_1\lambda_2\theta_2\theta_3 + a\lambda_1\lambda_2\theta_2 + b\lambda_1\lambda_2\theta_3 + \lambda_1\lambda_2\theta_2 + \lambda_1\lambda_2\theta_3 \neq 0\hspace{8mm}\\
&&c\lambda_1\lambda_3\theta_2 + c\lambda_3\theta_2 + \lambda_1\lambda_3\theta_2 + \lambda_1\theta_2 + \lambda_1 + \lambda_3 \neq 0\hspace{35mm}\\
&&d\lambda_1\lambda_3\theta_3 + d\lambda_3\theta_3 + \lambda_1\lambda_3\theta_3 + \lambda_1\theta_3 + \lambda_1 + \lambda_3 \neq 0\hspace{34mm}\\
&&c\lambda_1\lambda_3\theta_2\theta_3 + d\lambda_1\lambda_3\theta_2\theta_3 + c\lambda_3\theta_2 + d\lambda_3\theta_3 + \lambda_1\theta_2 + \lambda_1\theta_3 \neq 0\hspace{18mm}\\
&&c\lambda_1\lambda_3\theta_2\theta_3 + d\lambda_1\lambda_3\theta_2\theta_3 + c\lambda_1\lambda_3\theta_2 + d\lambda_1\lambda_3\theta_3 + \lambda_1\lambda_3\theta_2 + \lambda_1\lambda_3\theta_3 \neq 0\hspace{7mm}\\
&&a\lambda_2\lambda_3\theta_2 + c\lambda_2\lambda_3\theta_2 + a\lambda_2\theta_2 + c\lambda_3\theta_2 + \lambda_2 + \lambda_3 \neq 0\hspace{30mm}\\
&&b\lambda_2\lambda_3\theta_3 + d\lambda_2\lambda_3\theta_3 + b\lambda_2\theta_3 + d\lambda_3\theta_3 + \lambda_2 + \lambda_3 \neq 0\hspace{30mm}\\
&&bc\lambda_2\lambda_3\theta_2\theta_3 + ad\lambda_2\lambda_3\theta_2\theta_3 + a\lambda_2\theta_2 + c\lambda_3\theta_2 + b\lambda_2\theta_3 + d\lambda_3\theta_3 \neq 0\hspace{10mm}\\
&&bc\lambda_2\lambda_3\theta_2\theta_3 + ad\lambda_2\lambda_3\theta_2\theta_3 + a\lambda_2\lambda_3\theta_2 + c\lambda_2\lambda_3\theta_2 + b\lambda_2\lambda_3\theta_3 + d\lambda_2\lambda_3\theta_3 \neq 0\hspace{-1mm}\\
&&a\lambda_1\lambda_2\theta_2 + c\lambda_1\lambda_3\theta_2 + a\lambda_2\lambda_3\theta_2 + c\lambda_2\lambda_3\theta_2 + \lambda_1\lambda_2\theta_2 + \lambda_1\lambda_3\theta_2 \neq 0\hspace{12mm}\\
&&b\lambda_1\lambda_2\theta_3 + d\lambda_1\lambda_3\theta_3 + b\lambda_2\lambda_3\theta_3 + d\lambda_2\lambda_3\theta_3 + \lambda_1\lambda_2\theta_3 + \lambda_1\lambda_3\theta_3 \neq 0\hspace{12mm}\\
&&bc\lambda_2\lambda_3\theta_2\theta_3 + ad\lambda_2\lambda_3\theta_2\theta_3 + a\lambda_1\lambda_2\theta_2\theta_3 + b\lambda_1\lambda_2\theta_2\theta_3 + c\lambda_1\lambda_3\theta_2\theta_3\\&&\quad +
d\lambda_1\lambda_3\theta_2\theta_3 \neq 0.\hspace{-15mm}
\end{eqnarray*}

\end{appendices}

\medskip

\bibliographystyle{plain}
\bibliography{ref}

\end{document}